\documentclass[11pt]{article}
\usepackage[a4paper, margin=1in]{geometry} 
\usepackage{mathtools}
\usepackage[colorlinks, linkcolor=blue,citecolor = green,urlcolor=cyan]{hyperref}
\usepackage[framemethod=tikz]{mdframed}
\usepackage{wrapfig}

\usepackage{fullpage}
\usepackage[english]{babel}
\usepackage[utf8]{inputenc}
\usepackage{amsmath}
\usepackage{graphicx}
\usepackage{amsmath,amsthm,amssymb,epsfig,enumerate,graphicx}
\usepackage{boxedminipage2e}
\usepackage{algpseudocode}
\usepackage{algorithm}
\usepackage{algorithmicx}
\newtheorem{theorem}{Theorem}
\newtheorem{claim}[theorem]{Claim}
\newtheorem{lemma}[theorem]{Lemma}
\newtheorem{observation}[theorem]{Observation}
\newtheorem{definition}[theorem]{Definition}
\newtheorem{corollary}[theorem]{Corollary}
\newtheorem{fact}[theorem]{Fact}

\newcommand{\congest}{${\mathsf{CONGEST}}$}

\newcommand{\dist}{\mbox{\rm dist}}

\newcommand{\dilation}{\mbox{\tt d}}
\newcommand{\congestion}{\mbox{\tt c}}

\newcommand{\ApproxGraphCover}{\mathsf{ApproxCycleCover}}

\newcommand{\DistEdgeConn}{\mathsf{DistEdgeConnec}}
\newcommand{\MinCut}{\mathsf{MinCut}}

\def\dnsparagraph#1{\par\vspace{2pt}\noindent{\bf #1}.}

\def\cA{{\cal A}}
\def\cC{{\cal C}}

\newcommand{\poly}{\mathsf{poly}}

\renewcommand{\paragraph}[1]{\vspace{0.15cm}\noindent {\bf #1}}

\usepackage{wrapfig}
\usepackage{geometry}
  \usepackage{nth}
  \usepackage{intcalc}

\title{Small Cuts and Connectivity Certificates: \\A Fault Tolerant Approach\footnote{Department of Computer Science and
    Applied Mathematics, Weizmann Institute of Science, Rehovot 76100,
    Israel. Emails:
    \texttt{\{merav.parter\}@weizmann.ac.il}.}}
\author{Merav Parter \thanks{Supported
    in part by grants from the Israel Science Foundation grant no.\ 2084/18.}}

\begin{document}
\begin{titlepage}
\date{}
\maketitle

\begin{abstract}
We revisit classical connectivity problems in the \congest\ model of distributed computing.
By using techniques from fault tolerant network design, we show improved constructions, some of which are even ``local" (i.e., with $\widetilde{O}(1)$ rounds) for problems that are closely related to hard global problems (i.e., with a lower bound of $\Omega(Diam+\sqrt{n})$ rounds).

\paragraph{Distributed Minimum Cut:} 
Nanongkai and Su presented a randomized algorithm for computing a $(1+\epsilon)$-approximation of the minimum cut using $\widetilde{O}(D +\sqrt{n})$ rounds where $D$ is the diameter of the graph. For a sufficiently large minimum cut $\lambda=\Omega(\sqrt{n})$, this is tight due to Das Sarma et al. [FOCS'11], Ghaffari and Kuhn [DISC'13]. 
\begin{itemize}
\item{\textbf{Small Cuts:}}
A special setting that remains open is where the graph connectivity $\lambda$ is small (i.e., constant). The only lower bound for this case is $\Omega(D)$, with a matching bound known only for $\lambda \leq 2$ due to Pritchard and Thurimella [TALG'11]. Recently, Daga, Henzinger, Nanongkai and Saranurak [STOC'19] raised the open problem of computing the minimum cut in $\poly(D)$ rounds for any $\lambda=O(1)$. In this paper, we resolve this problem by presenting a surprisingly simple algorithm, that takes a completely different approach than the existing algorithms. Our algorithm has also the benefit that it computes \emph{all} minimum cuts in the graph, and naturally extends to \emph{vertex} cuts as well. At the heart of the algorithm is a graph sampling approach usually used in the context of fault tolerant (FT) design.

\item{\textbf{Deterministic Algorithms:}}
While the existing distributed minimum cut algorithms are randomized, our algorithm can be made deterministic within the same round complexity. To obtain this, we introduce a novel definition of universal sets along with their efficient computation. This allows us to derandomize the FT graph sampling technique, which might be of independent interest.

\item{\textbf{Computation of all Edge Connectivities:}} We also consider the more general task of computing  the edge connectivity of all the edges in the graph. In the output format, it is required that the endpoints $u,v$ of every edge $(u,v)$ learn the cardinality of the $u$-$v$ cut in the graph. We provide the first sublinear algorithm for this problem for the case of \emph{constant} connectivity values. Specifically, by using the recent notion of low-congestion cycle cover, combined with the sampling technique, we compute all edge connectivities in $\poly(D) \cdot 2^{O(\sqrt{\log n\log\log n})}$ rounds.
\end{itemize}
%

\paragraph{Sparse Certificates:} 
For an $n$-vertex graph $G$ and an integer $\lambda$, a $\lambda$-sparse certificate $H$ is a subgraph $H \subseteq G$ with $O(\lambda n)$ edges which is $\lambda$-connected iff $G$ is $\lambda$-connected. 
For $D$-diameter graphs, constructions of sparse certificates for $\lambda \in \{2,3\}$ have been provided by Thurimella [J. Alg. '97] and Dori [PODC'18] respectively using $\widetilde{O}(D)$ number of rounds.
The problem of devising such certificates with $o(D+\sqrt{n})$ rounds was left open by Dori [PODC'18] for any $\lambda\geq 4$.  Using connections to fault tolerant spanners, we considerably improve the round complexity for \emph{any} $\lambda \in [1,n]$ and $\epsilon \in (0,1)$, by showing a construction of $(1-\epsilon)\lambda$-sparse certificates with $O(\lambda n)$ edges using only $O(1/\epsilon^2 \cdot\log^{2+o(1)} n)$ rounds. 


\end{abstract}

\end{titlepage}

\vspace{-10pt}\section{Introduction}\vspace{-5pt}
The connectivity of a graph is one of the most fundamental concept in graph theory and network reliability. 
In this paper, we revisit some classical connectivity problems in the \congest\ model of distributed computing via the lens of \emph{fault tolerant} network design. 
We mainly focus on two problems:  exact computation of small\footnote{By small we mean of constant size.} edge (or vertex) cuts; and the computation of sparse connectivity certificates. Both of these problems have been studied thoroughly in the literature, and surprisingly still admit critically missing pieces. 
By using techniques from fault tolerant network design we considerably improve the state-of-the-art, as well as provide the first \emph{deterministic} distributed algorithms for these problems.
\vspace{-10pt}
\subsection{Small Cuts} \vspace{-5pt} In the distributed minimum cut problem, given a graph $G$ with edge connectivity $\lambda$, the goal is to identify at least one minimum cut, that is a collection of $\lambda$ edges whose removal disconnect the graph. In the output format, each vertex should learn at least one possible minimum (edge) cut. We start by providing a brief history for the problem\footnote{Although historically, the lower bound by Das-Sarma et al. \cite{sarma2012distributed} appeared before the upper bound algorithms, we reverse the order of presentation here.}. 

\vspace{-3pt}
\paragraph{A Brief History.}
\emph{(I) Upper Bounds: }The first non-trivial distributed algorithm for the minimum cut problem was given by Ghaffari and Kuhn \cite{ghaffari2013distributed}. They presented a randomized algorithm for computing a $(2+\epsilon)$ approximation of minimum cut using $\widetilde{O}_{\epsilon}(D+\sqrt{n})$ rounds, with high probability. 
Shortly after, Nanongkai and Su \cite{nanongkai2014almost} improved the approximation ratio to $(1+\epsilon)$ with roughly the same round complexity. Recently, Daga et. al \cite{DagaSTOC19} provided an exact algorithm with sublinear round complexity which improves up on the state of the art in the regime of large cuts (i.e., for $\lambda=n^{\Omega(1)}$).

\emph{(II) Lower Bounds: }
In their seminal paper, Das-Sarma et al. 
\cite{sarma2012distributed} presented a lower bound of $\widetilde{\Omega}(D+\sqrt{n})$ rounds for the computation of an $\alpha$-approximation of a weighted minimum cut which holds even for graphs with diameter $D=O(\log n)$. 
This lower bound applies only for weighted graphs with large weighted minimum cut of size $\Omega(\sqrt{n})$.
Ghaffari and Kuhn \cite{ghaffari2013distributed-Arxiv} extended this lower bounds in two ways. First, they considered 
a weaker setting for weighted minimum cuts where the edge weights correspond to capacities, and thus nodes can exchange $O(w \log n)$ bits over edges of weight $w$ in each round. They showed that even in this weaker model, the $\alpha$-approximation of minimum cut in $\lambda$-edge connected graphs with $\lambda=\Theta(\sqrt{n})$ and diameter $O(\log n)$ requires $\widetilde{\Omega}(\sqrt{n/(\alpha \cdot \lambda}))$ rounds\footnote{In the conference version of \cite{ghaffari2013distributed}, a lower bound of $\widetilde{\Omega}(D+\sqrt{n})$ was mistakenly claimed for \emph{any} $\lambda\geq 1$ and graphs with diameter $D=O(\log n)$. This was later on fixed in a modified arXiv version \cite{ghaffari2013distributed-Arxiv} and in Ghaffari's thesis \cite{GhaffariThesis17}.}. Observe that since in this construction $\lambda=\Theta(\sqrt{n})$, this lower bound can also be stated as $\Omega(\sqrt{\lambda})$ rather than $\Omega(D+\sqrt{n})$.  \indent In their second extension, Ghaffari and Kuhn attempted to capture also unweighted simple graphs. Here, they showed a lower bound of $\widetilde{\Omega}(D+\sqrt{n/(\alpha \cdot \lambda)})$ rounds, for any $\lambda\geq 1$ but only for graphs with diameter $D=1/\lambda \cdot \sqrt{n/(\alpha \cdot \lambda)}$. Again\footnote{Also here the conference version \cite{ghaffari2013distributed} mistakenly claimed that the lower bound works even for graphs with diameter $D=O(\log n)$, and this was fixed in \cite{ghaffari2013distributed-Arxiv,GhaffariThesis17}.}, with such a larger diameter, one can alternatively state this lower bound as $\Omega(\lambda \cdot D)$, rather than $\widetilde{\Omega}(D+\sqrt{n})$.  

\dnsparagraph{Computation of Small Cuts}
The conclusion from the above discussion is that we still do not have matching bounds for the distributed minimum cut problem in cases where either (i) the value of the weighted minimum cut is $o(\sqrt{n})$,  or (ii) the unweighted diameter is $o(\sqrt{n})$. 

As most real-world networks admit small cuts \cite{wang2003complex}, we are in particular intrigued by the complexity of computing the cuts in unweighted graphs with constant connectivity. Pritchard and Thurimella \cite{pritchard2011fast} showed an $O(D)$-round randomized algorithm for cut values up to $2$.
The problem of devising an $\poly(D)$ round algorithm for any constant $\lambda=O(1)$ was recently raised by Daga et al. \cite{DagaSTOC19}:
\begin{quote}
\emph{``A special case that deserves attention is when the graph connectivity is small. For example,
is there an algorithm that can check whether an unweighted network has connectivity at most $k$ in $\poly(k,D,\log n)$? ... Bounds in these forms are currently known only for $k \leq 2$."}
\end{quote}

We answer this question in the affirmative by presenting a $\poly(D)$-round algorithm for any constant connectivity $\lambda=O(1)$. This algorithm in fact computes all possible minimum cuts in $G$, in the sense that for any min-cut set $E'$ there is at least one vertex in the graph that knows $E'$. Turning to \emph{vertex} cuts, \cite{pritchard2011fast} showed an  $O(D+\Delta/\log n)$ round algorithm for computing the cut vertices. No exact algorithm is known for the case where the vertex connectivity is at least two. 
Our algorithm can be easily adapted to compute deterministically the (exact) vertex cuts in $\poly(D\cdot \Delta)$ rounds where $\Delta$ is the maximum degree. 

\vspace{-10pt}
\subsection{Sparse Connectivity Certificates}
For a given unweighted $n$-vertex $D$-diameter graph $G=(V,E)$ and integer $\lambda\geq 1$, a \emph{connectivity certificate} is a subgraph $H \subseteq G$ satisfying that it is $\lambda$-edge (or vertex) connected iff $G$ is $\lambda$-edge (or vertex) connected. The certificate is said to be \emph{sparse} if $H$ has $O(\lambda n)$ edges. 
Sparse certificates were introduced by Nagamochi and Ibaraki \cite{nagamochi1992linear}. 
Thurimella \cite{thurimella1997sub} gave the first distributed construction of $\lambda$-sparse certificates using $O(\lambda \cdot (D+\sqrt{n}))$ rounds in the \congest\ model. 
For $\lambda=2$, Censor-Hillel and Dory \cite{Censor-HillelD17} showed\footnote{\cite{Censor-HillelD17} studied the problem of the minimum $k$-edge-connected spanning subgraph ($k$-EECS), which for unweighted graphs implies the computation of connectivity certificates with a small number of edges.} the randomized construction of a certificate with $O(n)$ edges using $O(D)$ rounds. In \cite{Dory18}, Dory considered the case of $\lambda=3$, and showed the construction of a certificate with $O(n\log n)$ edges and $O(D \log^3 n)$ rounds. 
These algorithms are randomized and are based on the cycle space sampling technique of Pritchard and Thurimella \cite{pritchard2011fast}.  The problem of designing sparse certificates for any $\lambda \geq 4$ using $\widetilde{O}(D)$ rounds was left open therein \cite{Dory18}.

In this paper, we provide an easy solution for this problem which takes only $\widetilde{O}(\lambda)$ rounds for any $\lambda$. This is based on the observation that fault tolerant spanners are in fact sparse connectivity certificates. As a result we get that the problem of designing sparse certificates is \emph{local} rather than global (i.e., does not depend on the graph diameter). In the Our Results section we also improve the round complexity into  $\widetilde{O}(1)$ (i.e., independent on $\lambda$) by loosing a small factor in the approximation.

\subsection{Our Results} 
\paragraph{Distributed Computation of Small Minimum Cuts.}
We consider an unweighted $D$-diameter graph $G=(V,E)$ with edge connectivity $\lambda=O(1)$. We show a $\poly(D)$-round randomized algorithm to compute the minimum cut whose high level description can be stated in just few lines:  
Fix a vertex $s$ and apply $\poly(D)$ iterations, where in iteration $i$ we do as follows. (i) Sample a subgraph $G_i \subseteq G$ by adding each edge $e$ into $G_i$ independently with some fixed probability $p$. (ii) Compute a truncated BFS tree rooted at $s$ up to depth $O(\lambda D)$ in $G_i$, and (iii) let each vertex $t$ collect its $s$-$t$ path in this tree (if such exists). Finally, after applying this procedure for $\poly(D)$ iterations, each vertex $t$ computes locally the $s$-$t$ cut on the subgraph that it has collected.
The argument shows that every vertex $t$ that is separated from $s$ by some minimum cut $E'$, can compute this set of edges w.h.p.
\begin{mdframed}[hidealllines=false,backgroundcolor=white!25]
\vspace{-10pt}
\begin{theorem}\label{thm:min-cut-rand} Let $G$ be an $\lambda=O(1)$ connected $D$-diam graph and max degree $\Delta$.
There exists a randomized minimum cut algorithm that runs in $\poly(D)$ rounds. In addition, with a small modification it computes the minimum vertex cut in $\poly(D\cdot \Delta)$ rounds.
\vspace{-3pt}
\end{theorem}
\end{mdframed}
The algorithm is in fact stronger. Every vertex $t$ also learns a collection $(\lambda-1)$ edge disjoint paths from $s$ (i.e., an integral flow from $s$). In addition, we do not compute only one minimum cut but rather for each minimum cut in $G$, there is at least one vertex that learns it. 
%


\paragraph{Deterministic Computation of Small Cuts.}
So-far, the distributed minimum cut computation was inherently randomized. 
The randomized component of the algorithm of Thm. \ref{thm:min-cut-rand} is in the initial graph sampling in each iteration. To derandomize it, we introduce a new variant of \emph{universal-sets}. We use this notion to explicitly compute, in polynomial time, a collection of $\poly(D)$ subgraphs $G_1,\ldots, G_k$ that have the same key properties as those obtained by the sampling approach. The polynomial computation is done \emph{locally} at each vertex and thus does not effect the round complexity.
\begin{mdframed}[hidealllines=false,backgroundcolor=white!25]
\vspace{-8pt}
\begin{theorem}\label{thm:min-cut-det} 
One can compute small cuts deterministically in $\poly(D)$ rounds.
\vspace{-3pt}
\end{theorem}
\end{mdframed}
This derandomization technique can be used to derandomize all other algorithms that are based on the fault tolerant (FT) sampling technique (e.g., \cite{dinitz2011fault},\cite{weimann2013replacement}), and it is therefore of interest also for centralized algorithms. Independently of our work, Alon, Chechik and Cohen \cite{AlonCC19} also studied the derandomization of algorithms that are based on the FT-sampling approach, their solution is different than ours.

For a summary on the computation of small cuts, see Table \ref{table-results-cuts}. 
\begin{table}[h]
\centering
\begin{tabular}{|c|c|c|c|}
	\hline & \textbf{Min-Cut Value $\lambda$} 	&  \textbf{\#Rounds} & \textbf{Type}  \\  \hline
	{\begin{tabular}[c]{@{}c@{}}
			Pritchard \& Thurimella \cite{pritchard2011fast}
	\end{tabular}}
	& $2$ Edges & $O(D)$ & {Rand.}   \\ \hline
	{\begin{tabular}[c]{@{}c@{}}
				Pritchard \& Thurimella \cite{pritchard2011fast}
		\end{tabular}}
	& \begin{tabular}[c]{@{}c@{}}
		$1$ Vertex
	 \end{tabular}
	 & $O(D+\Delta)$ &	{Rand.} \\  \hline
  {This Work} & $O(1)$ Edges & $\poly(D)$ & {Det.} \\  \hline
  {This Work} & $O(1)$ Vertices & $\poly(D \cdot \Delta)$ & {Det.} \\  \hline
\end{tabular}
\caption{\label{table-results-cuts}State of the art results for exact distributed computation of small cuts.}
\end{table}

\paragraph{Computation of Edge-Connectivities.}
We then turn to consider the more general task of computing the edge connectivity of all graph edges, up to some constant bound $\lambda=O(1)$. For an edge $e=(u,v)$, the \emph{edge connectivity} of $e$ is the size of the $u$-$v$ minimum (edge) cut in $G$. In the output format for each edge $e=(u,v)$, its endpoints are required to learn the edge connectivity of $e$. 
Exact computation of all edge connectivites has been previously known only $\lambda \leq 2$ due to Pritchard and Thurimella \cite{pritchard2011fast}. They gave randomized algorithms for the case of $\lambda=1,2$ with round complexities of $O(D)$ and $O(D+\sqrt{n}\log^* n)$, respectively.

In this paper, we again take a completely different approach and show a \emph{deterministic} algorithm with $\poly(D)\cdot 2^{\sqrt{\log n\log\log n}}$ rounds for computing all edge connectivities up to constant value of $\lambda=O(1)$. Our algorithm is based on two tools: (1) \emph{low-congestion cycle cover} \cite{parter-yogevSODAa} and their distributed computation \cite{parter-yogevCycles} ; and (2) the derandomization of the FT-sampling approach.
%
\begin{mdframed}[hidealllines=false,backgroundcolor=white!25]
\vspace{-8pt}
\begin{theorem}\label{thm:edge-conn} For every $D$-diameter $n$-vertex graph $G=(V,E)$, w.h.p., the edge connectivity of all graphs edges up to $\lambda=O(1)$ can be computed in $\poly(D)\cdot 2^{O(\sqrt{\log n})}$ rounds. This algorithm can also be derandomized using $\poly(D)\cdot 2^{O(\sqrt{\log n\log\log n})}$ rounds.
\end{theorem}
\end{mdframed}
%
\vspace{-5pt}
\paragraph{Sparse Connectivity Certificates.} In the second part of the paper we consider the related problem of computing connectivity certificates. 
We first show that by a direct application of fault tolerant spanners, one can compute a $\lambda$-edge connectivity certificate with $O(\lambda n)$ edges using $\widetilde{O}(\lambda)$ rounds. This considerably improves and extends up on the previous constructions with $O(D)$ rounds that were limited \emph{only} for $\lambda \in \{2,3\}$. 
\begin{mdframed}[hidealllines=false,backgroundcolor=white!25]
\vspace{-8pt}
\begin{lemma} For every $\lambda \in \mathbb{N}_{\geq 1}$, there is a randomized algorithm that computes a $\lambda$ connectivity certificate with $O(\lambda n)$ edges in $O(\lambda \log^{1+o(1)} n)$ rounds, with high probability.
\end{lemma}
\end{mdframed}
By plugging in the recent \emph{deterministic} spanner construction of \cite{GhaffariK18}, one can compute $\lambda$-edge connectivity certificate deterministically with $\widetilde{O}(\lambda \cdot n)$ edges and $\lambda \cdot 2^{O(\sqrt{\log n})}$ rounds\footnote{Combining the recent result of \cite{detND-Arxiv} with \cite{GhaffariK18} seems to improve the deterministic spanner construction to $\poly \log n$ rounds, and thus provide $\widetilde{O}(\lambda)$-round algorithm for $\lambda$-sparse certificates.}. This answers the open problem raised by Dory \cite{Dory18} concerning the existence of efficient deterministic constructions of connectivity certificates.  

To avoid the dependency in $\lambda$ in the round complexity, we use the well known Karger's edge-sampling technique, and show:
\begin{mdframed}[hidealllines=false,backgroundcolor=white!25]
\vspace{-8pt}
\begin{lemma} For every $\lambda$-connected graph and $\epsilon \in (0,1)$, there is a randomized distributed algorithm that computes a $(1-\epsilon)\lambda$ connectivity certificate with $O(\lambda n)$ edges in $O(1/\epsilon^2 \cdot\log^{2+o(1)} n)$ rounds, with high probability. 
\vspace{-3pt}
\end{lemma}
\end{mdframed}
Table \ref{table-results-certificates-new} summarizes the state of the art. Note that if one uses fault tolerant spanners resilient for \emph{vertex} faults, we get $\widetilde{O}(\lambda)$-round algorithm for computing $\lambda$-vertex-certificates\footnote{I.e., a subgraph $H \subseteq G$ satisfying that $H$ is $\lambda$-vertex connected iff $G$ is $\lambda$-vertex connected.} with $O(\lambda^2 n)$ edges. For clarity of presentation, we mainly focus on the edge-connectivity certificates, but our results naturally extend to the vertex case as well. 

\begin{table}[h]
\centering
\begin{tabular}{|c|c|c|c|c|}
	\hline & \textbf{Edge Connectivity $\lambda$} 	& \textbf{Certificate Size} & \textbf{\#Rounds} & \textbf{Type}  \\  \hline
	{\begin{tabular}[c]{@{}c@{}}
			Thurimella  \cite{thurimella1997sub}
	\end{tabular}}
	& Any & $\lambda\cdot n$ & $\widetilde{O}(\lambda \cdot (D+\sqrt{n}))$ & {Det.}   \\ \hline
	{\begin{tabular}[c]{@{}c@{}}
				Pritchard \& Thurimella \cite{pritchard2011fast}
		\end{tabular}}
	& \begin{tabular}[c]{@{}c@{}}
		$2$
	 \end{tabular}
	 & $2n$ & $O(D)$ &	{Rand.} \\  \hline
		{\begin{tabular}[c]{@{}c@{}}
				Dory \cite{Dory18}
		\end{tabular}}
	& \begin{tabular}[c]{@{}c@{}}
		$3$
	 \end{tabular}
	 & $O(n\log n)$ & $O(D\cdot \log^3 n)$ &	{Rand.} \\  \hline	
  {This Work} & Any & $O(\lambda n)$ & $O(\lambda \cdot \log^{1+o(1)}n)$ & {Rand.} \\  \hline
  {This Work}	& \begin{tabular}[c]{@{}c@{}}
  	Approx. $(1-\epsilon)$
  \end{tabular}  & $O(\lambda \cdot n)$ & $O(\log^{2+o(1)}n)$ & {Rand.} \\  \hline
	{This Work}	& Any  & $\widetilde{O}(\lambda \cdot n )$ & $\lambda \cdot 2^{O(\sqrt{\log n\log\log n})}$ & {Det.} \\  \hline
\end{tabular}
\caption{State of the art result for distributed computation of sparse connectivity certificates. \label{table-results-certificates-new}}
\end{table}

\paragraph{Graph Notation.} For a subgraph $G'\subseteq G$ and $u,v \in V(G')$, let $\pi(u,v,G')$ be the unique\footnote{Ties are broken is a consistent manner.} shortest-path in $G'$. When $G'$ is clear from the context, we may omit it and simply write 
$\pi(u,v)$. For $u,v \in G$, let $\dist(u,v,G)$ be the length of the shortest $u$-$v$ path in $G$.
For a vertex pair $s,t$ and subgraph $G' \subseteq G$, let $\lambda(s,t,G')$ be the $s$-$t$ cut in $G'$. 
\dnsparagraph{The Communication Model}
We use a standard message passing model, the
\congest\ model \cite{Peleg:2000}, where the execution proceeds in synchronous 
rounds and in each round, each node can send a message of size $O(\log n)$ to 
each of its neighbors. 

\vspace{-15pt}\section{Exact Computation of Small Cuts} 
Throughout, we consider unweighted multigraphs with diameter $D$, and (edge or vertex) connectivity at most $\lambda=O(1)$. Before presenting the algorithm we start by considering the following  simpler task.

\dnsparagraph{Warm Up: Cut Verification}
In the cut verification problem, one is given a subset of edges $E'$ where $|E'|\leq \lambda$, it is then required to test if $G \setminus E'$ is connected. As we will see there is a simple algorithm for this problem which is based on the following key lemma.
\vspace{-5pt}
\begin{lemma}\label{lem:bounded-length-connected}
Consider a $D$-diameter unweighted graph $G=(V,E)$ with maximum degree $\Delta$. 
(1) If $u,v \in V$ are $\lambda$-edge connected\footnote{Note that we do not require the graph $G$ to be $\lambda$ connected.} (i.e., the $u$-$v$ cut is at least $\lambda$) then $\dist(u,v, G \setminus F)\leq c\cdot\lambda \cdot D$ for every edge sequence $F \subseteq E$, $|F|\leq \lambda-1$ for some constant $c$. \\
(2) If $u,v \in V$ are $\lambda$-vertex connected then $\dist(u,v, G \setminus F)\leq c \cdot \lambda \cdot\Delta \cdot D$ for every vertex sequence $F \subseteq V$, $|F|\leq \lambda-1$. 
\end{lemma}
\begin{proof}
Let $T$ be an arbitrary BFS tree in $G$ of diameter $O(D)$. We begin with (1). Fix a set of faults $F \subseteq E$, $|F| \leq \lambda-1$, and let $P_{u,v,F}$ be the $u$-$v$ shortest path in $G \setminus F$. Since $u$ and $v$ are $\lambda$-edge connected such path $P_{u,v,F}$ exists. We next bound the length of $P_{u,v,F}$. Consider the forest $T \setminus F$ which has at most $\lambda$ connected components $C_1,\ldots, C_\ell$. We mark each vertex on $P_{u,v,F}$ with its component ID in the forest $T \setminus F$. Note that all vertices in the same component are connected by a path of length $O(D)$ in $G \setminus F$. We can then traverse the path $P_{u,v,F}$ from $u$ and jump to the last vertex $u_1$ on the path (close-most to $v$) that belongs to the component of $u$. The length of this sub-path is $O(D)$ and using one connecting edge on $P_{u,v,F}$, we move to a vertex belonging to a new component in $T \setminus F$. Overall the path $P_{u,v,F}$ can be covered by $\lambda$ path segments $P_1,\ldots, P_{\ell}$ such that the endpoints of each segments are in the \emph{same} component in $T \setminus F$, each neighboring segments $P_i$ and $P_{i+1}$ are connected by an edge from $P_{u,v,F}$. Since each $|P_i|=O(D)$, we get that $|P_{u,v,F}|=O(\lambda \cdot D)$. 
Claim (2) follows the exact same argument with the only distinction is that when a vertex fails, the BFS tree might break up into $\Delta+1$ components. Thus, for a subset $F \subseteq V$ of vertices with $|F|\leq \lambda-1$, the tree $T$ breaks into $O(\lambda \cdot \Delta)$ components.
\end{proof}
This lemma immediately implies an $O(\lambda D)$-round solution for the cut verification task: build a BFS tree $T$ from an arbitrary source up to depth $O(\lambda \cdot D)$. Then $T$ is a spanning tree iff $E'$ does not disconnect the graph. \vspace{-7pt}
\begin{corollary}\label{cor:cut-verification}[Cut Verification]
Given a set of $\lambda$ edges $E'$. One can test if $E'$ is a cut in $G$ using $O(\lambda \cdot D)$ rounds. 
\end{corollary}

%
%
%
%
%
The following definition is useful for the description and analysis of our algorithm.
\begin{definition}[$(s,t)$ connectivity certificate]
Given a graph $G$ with minimum cut $\lambda$ and a pair of vertices $s$ and $t$, the $(s,t)$ \emph{connectivity certificate} is a subgraph $G_{s,t} \subseteq G$ satisfying that $s$ and $t$ are $\lambda$-connected in $G_{s,t}$ iff they are $\lambda$-connected in $G$.
\end{definition}\vspace{-7pt}
Whereas a-priori the size of the $s$-$t$ connectivity certificate, measured by the number of edges, might be $\Omega(n)$, as will show later on, it is in fact bounded by $(\lambda D)^{O(\lambda)}$, hence $\poly(D)$ for $\lambda=O(1)$. With this definition in mind, we are now ready to present the minimum cut algorithm.

\paragraph{A $\poly(D)$-Round Randomized Algorithm.}
The algorithm has two phases. In the first phase, every vertex $t$ computes its $(s,t)$ certificate subgraph $G_{s,t}$ w.r.t. a given fixed source $s$. In the second phase, each vertex $y$ locally computes its $s$-$t$ cut in the subgraph $G_{s,t}$, and one of the output $\lambda$-size cut is broadcast to the entire network. 
%
Throughout, we assume w.l.o.g. that the value of the minimum cut $\lambda$ is known, since $\lambda=O(1)$ this assumption can be easily removed.

The first phase has $\ell=O((\lambda D)^{2\lambda})$ iterations, or \emph{experiments}. 
In each iteration $i$, the algorithm samples a subgraph $G_i$ by including each edge $e \in G$ into $G_i$ independently with probability $p=1-1/(c(\lambda D))$ for some constant $c$ 
(taken from Lemma \ref{lem:bounded-length-connected}). 
For a source vertex $s$ (which is fixed in all iterations), a (truncated) BFS tree $B_i$ rooted in $s$ is computed in $G_i$ up to depth $c\cdot \lambda \cdot D$.
Next, every vertex in $B_i$ learns its tree path from $s$ by pipelining these edges downward the tree. This completes the description of an iteration. Let $G_{s,t}=\bigcup_{i=1}^{\ell} \pi(s,t,B_i)$.

In the second phase, every vertex $t$ locally computes its $s$-$t$ cut in the subgraph $G_{s,t}$.
The edges of the minimum cut are those obtained by one of the vertices $t$ whose $s$-$t$ connectivity in $G_{s,t}$ is at most $\lambda$. 

\dnsparagraph{Correctness} 
For the correctness of the algorithm it will be sufficient to show that w.h.p. $G_{s,t}$ is an $s$-$t$ connectivity certificate for every $t \in V$.
\begin{claim}\label{cl:succprobcut}
For every $t$, w.h.p., $G_{s,t}$ is an $s$-$t$ connectivity certificate.
\end{claim}
\begin{proof}
Since $G_{s,t} \subseteq G$, it is sufficient to show that $s$ and $t$ are connected in $G_{s,t} \setminus F$ for any subset $F \subset E$, $|F|\leq \lambda$ satisfying that $s$ and $t$ are connected in $G \setminus F$. 
Fix such a triplet $\langle s,t,F \rangle$ where $s$ and $t$ are connected in $G \setminus F$. 
An iteration $i$ is \emph{successful} for $\langle s,t, F \rangle$ if 
$$\pi(s,t,G\setminus F)\subseteq G_i \mbox{~~and~~} F \cap G_i=\emptyset.$$
Note that if iteration $i$ is successful for $\langle s,t, F \rangle$, then the truncated BFS tree $B_i$ contains $t$ as $\dist(s,t,G_i)= |\pi(s,t,G\setminus F)|\leq c\cdot \lambda \cdot D$, where the last inequality is due to Lemma \ref{lem:bounded-length-connected}(1). In addition, since $F \cap G_i=\emptyset$, it also holds that $\pi(s,t,B_i)\subseteq G_{s,t} \setminus F$.

It remains to show that with probability at least $1-1/n^{\Omega(\lambda)}$, every triplet $\langle s,t, F \rangle$ where $s$ and $t$ are connected in $G \setminus F$, has at least one successful iteration. The claim will then follow by applying the union bound over all $n^{2\lambda}$ triplets. Recall that each edge is sampled into $G_i$ independently with probability $p$. Thus the probability that iteration $i$ is successful for $\langle s,t, F \rangle$ is at least: 
$$q=p^{(c\cdot \lambda D)}\cdot (1-p)^\lambda=1/(\lambda\cdot D)^{\lambda}.$$ 
Since there are $\ell$ independent experiments, the probability that all of them fail is $(1-q)^{\ell}\leq 1/n^{\Omega(\lambda)}$, the claim follows.
\end{proof}
Finally, let $t$ be a vertex such that $s$ and $t$ are not $(\lambda+1)$-connected in $G$. Thus, by the lemma above, $\lambda(s,t,G_{s,t})\leq \lambda$ and the minimum cut computation applied locally by vertex $t$ in $G_{s,t}$ outputs a subset $F$ of at most $\lambda$ edges. We claim that w.h.p. $s$ and $t$ are also not connected in $G \setminus F$. Assume otherwise, then by Lemma \ref{lem:bounded-length-connected}(1),  $\dist(s,t,G \setminus F)\leq c\cdot \lambda \cdot D$, thus by the argument above, w.h.p., there is an iteration $i$ in which an $s$-$t$ path that does not go through $F$ is taken into $G_{s,t}$, leading to a contradiction that $s$ and $t$ are disconnected in $G_{s,t}\setminus F$.  

%
\vspace{-7pt}
\begin{corollary}
For every $D$-diameter unweighted graph $G=(V,E)$, $\lambda\geq 1$ and vertex pair $s,t \in V$, there exists an $(s,t)$ certificate $G_{s,t} \subseteq G$ of size $(\lambda D)^{O(\lambda)}$.
\end{corollary}\vspace{-7pt}
\def\APPENDFIGCERT{
\begin{figure} 
	\begin{center}
		\includegraphics[scale=0.4]{st-certificate.pdf}
	\end{center}
	\caption{An illustration of $s$-$t$ certificate of size $\Omega(D^{\lambda})$. Shown is a pair of neighbors $s$ and $t$. Every edge in this graph lies on a cycle of length $\Theta(D)$. Thus, the removal, of each edge increases the length of the given shortest-path by an additive factor of $D$. Note that in this example, $\dist(s,t, G \setminus F)=\Theta(D \cdot \lambda)$ for every $|F|\leq \lambda$. However, in the collection of $\lambda$-edge disjoint $s$-$t$ paths the length of each path is $\Theta(D^{\lambda})$. The sparse certificate $G_{s,t} \subseteq G$ of $s$ and $t$ must include all those edges, leading to the size bound of $D^{O(\lambda})$ edges.}
\end{figure} \label{fig:st-certificate}
}

\dnsparagraph{Round Complexity} Each of the $\ell$ iterations takes $O(\lambda \cdot D)$ rounds for computing the truncated BFS tree. Learning the edges along the tree path from the root also takes $O(\lambda \cdot D)$ rounds via pipeline, thus overall the round complexity is $(\lambda \cdot D)^{(c+2)\lambda}=\poly(D)$ for $\lambda=O(1)$. 

\dnsparagraph{Extension to Vertex Cuts} The algorithm for computing vertex cuts is almost identical and requires  minor adaptations. First, instead of having a single source vertex $s$, we will pick $\lambda+1$ arbitrary sources $s_1,\ldots, s_{\lambda+1}$ and will run an algorithm, which is very similar to the one described above, with respect to each source $s_i$. Note that since the vertex cut has size $\lambda$, then there is at least one vertex cut $V' \subset V$ of size $\lambda$ that does not contain at least one of the sources $s_i$. In such a case, our algorithm will find the cut $V'$ when running the below mentioned algorithm w.r.t the source $s_i$. 

The algorithm for each source $s_i$ works in iterations, where each iteration $j$ samples into a subgraph $G_j$ a collection of vertices rather than edges. That is, the subgraph $G_j$ is defined by taking the induced graph on a sample of vertices, where each vertex gets sampled independently with probability 
$p'=(1-1/(c\lambda \cdot \Delta \cdot D))$ (the constant $c$ is taken from Lemma \ref{lem:bounded-length-connected}(2)). Then a BFS tree $B_j$ rooted at $s_i$ is computed in $G_j$ up to depth $c\lambda \cdot \Delta \cdot D$. Every vertex $v \in B_j$ collects its path from the root. Let $G_{s_i,t}$ be the union of all paths collected for each vertex $t$. In the second phase, the vertex $t$ computes locally the $s_i$-$t$ vertex-cut in $G_{s_i,t}$. The analysis is then identical to that of the edge case, where in particular, we get that w.h.p. $G_{s_i,t}$ is the vertex-connectivity certificate for every $t$.

\vspace{-10pt}
\subsection{Deterministic Min-Cut Algorithms via Universal Sets}\label{sec:det-cut}\vspace{-5pt}
Our goal in this section is to derandomize the FT-sampling technique by locally computing explicitly (at each node) a \emph{small} family of graphs $\mathcal{G}=\{G_i \subseteq G\}$ such that in iteration $i$, the vertices will apply the computation on the graph $G_i$ in the same manner as in the randomized algorithm. Here, however, the graph $G_i$ is not sampled but rather computed locally by all the vertices. The family of subgraphs $\mathcal{G}$ is required to satisfy the following crucial property for $a=c \cdot D \cdot \lambda$ and $b=\lambda$:
\begin{mdframed}[hidealllines=false,backgroundcolor=white!25]
For every two disjoint subsets of edges $A,B$ with $|A|\leq a$ and $|B|\leq b$, there exists a subgraph $G_i \in \mathcal{G}$ satisfying that:
\begin{equation}\label{eq:det-sample}
A \subseteq G_i\mbox{~~and~~} B \cap G_i=\emptyset~.
\end{equation}
\end{mdframed}
\vspace{-5pt} In our algorithms, the subset $B$ corresponds to a set of edge faults, and $A$ corresponds to an $s$-$t$ shortest path in $G \setminus B$. Thus, $|B| \leq \lambda$ and by Lemma \ref{lem:bounded-length-connected}, $|A|=O(\lambda \cdot D)$. 
We begin with the following observation that follows by the probabilistic method.
\begin{lemma}\label{lem:prob-approach-det-sample}
There exists a family of graphs $\mathcal{G}=\{G_i \subseteq G\}$ of size $O(a^{b+1}\cdot \log n)$ that satisfies Eq. (\ref{eq:det-sample}) for every disjoint $A,B\subseteq E$ with $|A|\leq a$ and $|B|\leq b$. 
\end{lemma}
\begin{proof}
We will show that a random family $\mathcal{G}_R$ with $\ell=O(a^{b+1}\cdot \log n)$ subgraphs satisfies Eq. (\ref{eq:det-sample}) with non-zero probability. Each subgraph $G_i$ in $\mathcal{G}_R$ is computed by sampling each edge in $G$ into $G_i$ with probability of $p=(1-1/a)$. 

The probability that $G_i$ satisfies Eq. (\ref{eq:det-sample}) for a fixed set $A$ and $B$ of size at most $a$ and $b$ (respectively) is $q=p^a\cdot (1-p)^b=1/(e \cdot a^b)$. The probability that none of the subgraph $G_i$ satisfy   Eq. (\ref{eq:det-sample}) for $A,B$ is at most $(1-q)^{c \cdot a^{b+1}\cdot \log n}\leq 1/n^{3a}$ for some constant $c\geq 5e$. Thus, by applying the union bound over all $n^{2a}$ possible subsets of $A,B$, we get that $\mathcal{G}_R$ satisfies Eq. (\ref{eq:det-sample}) for all subsets with positive probability. The lemma follows.
\end{proof}
Lemma \ref{lem:prob-approach-det-sample} already implies a deterministic minimum cut algorithm with $\poly(D)$ rounds, in case where nodes are allowed to perform \emph{unbounded} local computation. Specifically, let every node compute locally, in a brute force manner, the family of graphs $\mathcal{G}=\{G_i \subseteq G\}$ of size $a^{b+1}\cdot \log n$. In each iteration $i$ of the minimum-cut computation, nodes will use the graph $G_i \in \mathcal{G}$ to compute the truncated BFS tree, and collect their tree paths in these trees. Although the \congest\ model does allow for an unbounded local computation, it is still quite undesirable. To avoid this, we next describe an explicit \emph{polynomial} construction of the graph family $\mathcal{G}$. This explicit construction is based on stating our requirements in the language of universal sets. 

\dnsparagraph{Universal Sets} A family of sets $\mathcal{S}=\{S \subseteq [n]\}$ is $(n,k)$-\emph{universal} if every subset $S' \subseteq [n]$ of $|S'|=k$ elements is \emph{shattered} by $\mathcal{S}$. That is, for each of the $2^k$ subsets $S'' \subseteq S'$ there exists a set $S \in \mathcal{S}$ such that $S'\cap S=S''$. Using linear codes, one can compute $(n,k)$-universal sets with $n^{O(k)}$ subsets. Alon \cite{alon1986explicit} showed an explicit construction of size $2^{O(k^4)}\log n$ using the Justesen-type codes  constructed by Friedman \cite{friedman1984constructing}. 
In our context, the parameter $n$ corresponds to the number of graph edges, and each subset is a subgraph. The parameter $k$ corresponds to the bound on the length of the path which is $O(\lambda \cdot D)$. Using the existing constructions lead to a family with $2^{\lambda \cdot D}$ subgraphs which is unfortunately super-linear, already for graphs of logarithmic diameter. 

\dnsparagraph{A New Variant of Universal Sets} We define a more relaxed variant of universal sets, for which a considerably improved size bounds can be obtained. In particular, for our purposes it is not really needed to fully shatter subsets of size $k$. Instead, for every set $S'$ of $k$ elements we would like that for every small subset $S''\subseteq S$, $|S''|\leq b$ (which plays the role of the faulty edges),  there will be a set $S$ in the family satisfying that $S' \cap S=S'\setminus S''$. 
We call this variant FT-universal sets, formally defined as follows.
%
\begin{definition}[FT-Universal Sets]
For integers $n,a,b$ where $a\leq b \leq n$, a family of sets $\mathcal{S}=\{ S \subset [1,n]\}$ is $(n,a,b)$-\emph{universal} if for every two disjoint subsets $A \subset [1,n]$ and $B \subset [1,n]$ where $|A|\leq a$ and $|B|\leq b$, there exists a set $S \in \mathcal{S}$ such that (1) $A \subseteq S$ and (2) $B \cap S=\emptyset$.
%
\end{definition}\vspace{-3pt}
Our goal is compute a family of $(n,a,b)$-universal sets of cardinality $\widetilde{O}(a^{2b+O(1)})$ in time $\poly(a^{O(1)+2b} \cdot n)$. 
Towards that goal we will use the notion of perfect hash functions.
\begin{definition}[Perfect Hash Functions]
For integers $n$ and $k<n$ a family of hash functions $\mathcal{H}=\{h: [n] \to [\ell]\}$ is \emph{prefect} if for every subset $S \subset [1,n]$ for $|S| \leq k$, there exists a function $h \in  \mathcal{H}$ such that
$h(i) \neq h(j), ~~~\forall i, j \in S, i \neq j.$
\end{definition}
\vspace{-10pt}
\begin{definition}[Almost Pairwise Independence]\label{def:almost}
A family of functions $\mathcal{H}$ mapping domain $[n]$ to range $[m]$ is $\epsilon$-almost pairwise independent if for every $x_1 \neq x_2 \in [n]$, $y_1,y_2 \in [m]$, we have:
$\Pr[h(x_1)=y_1 \mbox{~and~} h(x_2)=y_2]\leq (1+\epsilon)/m^2$.
\end{definition}
\begin{fact}[\cite{TCS-010,celis2013balls}]\label{fact:pihash}
For every $\alpha,\beta \in \mathbb{N}$ and constant $\epsilon \in (0,1)$, one can compute in $\poly(\alpha \cdot 2^{\beta}\cdot 1/\epsilon)$ an explicit family of $\epsilon$-almost pairwise independent hash functions $\mathcal{H}_{\alpha,\beta}=\{h: \{0,1\}^\alpha \to \{0,1\}^\beta\}$ that contains $\poly(\alpha \cdot 2^{\beta}\cdot 1/\epsilon)$ functions. 
\end{fact}

We next show how to compute a family of $(n,k)$-perfect hash functions in polynomial time.
\begin{claim}\label{cl:small-perfect}
One can compute a family of $(n,k)$-perfect hash functions $\mathcal{H}=\{h: [n] \to [2k^2]\}$ of cardinality $\poly(k\log n)$ in time $\poly(n,k)$.
\end{claim}
\begin{proof}
We use Fact \ref{fact:pihash} with $\alpha=\log n$, $\beta=\log(2k^2)$, and $\epsilon=0.1$, to get an $\epsilon$-almost pairwise independent hash function family $\mathcal{H}=\{h: \{0,1\}^\alpha \to \{0,1\}^\beta\}$. 
We now show that this family is perfect for subsets $S \in [1,n]$ of cardinality at most $k$.
Fix a subset $S \in [1,n]$, $|S|\leq k$. By definition, for every $x_1, x_2 \in S$ and $y_1, y_2 \in [2k^2]$,
$$\Pr_{h \in \mathcal{H}}[h(x_1)=y_1 \mbox{~and~} h(x_2)=y_2]\leq 1.1/(4 k^4).$$
Thus, the probability that a uniformly chosen random function $h \in \mathcal{H}$ collides on $S$ is 
\begin{eqnarray}
\sum_{x_1 \neq x_2 \in S}\Pr_{h \in \mathcal{H}}[h(x_1)=h(x_2)]&\leq& k^2 \cdot \max_{x_1 \neq x_2 \in S}\Pr_{h \in \mathcal{H}}[h(x_1)=h(x_2)]
\\& =& \nonumber
k^2\cdot \max_{x_1 \neq x_2 \in S}\sum_{y \in [2k^2]}\Pr[h(x_1)=h(x_2)=y]\leq 0.6~, \nonumber
\end{eqnarray}
by using the fact that $\Pr_{h \in \mathcal{H}}[h(x_1)=h(x_2)=y]\leq 1.1/(4k^4)$ (
see Def. \ref{def:almost}). We get that there exists $h' \in \mathcal{H}$ that has no collisions on $S$. As this holds for every $S$, the claim follows.
\end{proof}
Equipped with the polynomial construction of families of $(n,k)$-perfect hash functions, we next show how to compute our universal sets in polynomial time.
\begin{lemma}[Small Universal Sets]\label{lem:universalsets}
For every set of integers $b<a<n$, one can compute in $\poly(n,a^{b})$, a family of universal sets $\mathcal{S}_{n,a,b}$ of cardinality $\widetilde{O}((4a)^{O(1)+2b})$.
\end{lemma}
\begin{proof}
Set $k=a+b$. We will use Claim \ref{cl:small-perfect} to compute an $(n,k)$-perfect family of hash functions $\mathcal{H}=\{h: [n] \to [2k^2]\}$. 
For every $h \in \mathcal{H}$ and for every subset $i_1,\ldots, i_b \in [1,2k^2]$, define:
$$S_{h,i_1,i_2,\ldots, i_b}=\{ \ell \in [n] ~\mid~ h(\ell) \notin \{i_1,i_2,\ldots, i_b\}\}~.$$
Overall, $\mathcal{S}_{n,a,b}=\{S_{h,i_1,i_2,\ldots, i_b} ~\mid~ h \in \mathcal{H}, i_1,i_2,\ldots, i_b \in [1,2k^2]\}$. 

The size of $\mathcal{S}_{n,a,b}$ is bounded by $O(|\mathcal{H}|\cdot (2k)^{2b})=\widetilde{O}((4a)^{O(1)+2b})$ as desired. We now show that $\mathcal{S}_{n,a,b}$ is indeed a family of universal sets for $n,a,b$. 
Since $\mathcal{H}$ is an $(n,k)$ perfect family of hash functions, for every two disjoint subsets $A,B \subset [n]$, $|A|\leq a$ and $|B|\leq b$, there exists a function $h$ that does not collide on $C=A \cup B$ (since $|C|\leq k$). That is, there exists a function $h \in \mathcal{H}$ such that $h(i) \neq h(j)$ for every $i,j \in C$, $i \neq j$. 
Thus, letting $B=\{s_1,\ldots, s_b\}$ and $i_1=h(s_1),\ldots, i_b=h(s_b)$, we have that $h(s'_j)\notin \{i_1,\ldots, i_b\}$ for every $s'_j \in A$. Therefore, the subset $S_{h,i_1,i_2,\ldots, i_b}$ satisfies that 
$A  \subseteq S_{h,i_1,i_2,\ldots, i_b}$ and $B \cap  S_{h,i_1,i_2,\ldots, i_b}=\emptyset$. 
\end{proof}

\paragraph{Deterministic Min-Cut Algorithm.}
Finally, we describe how to use FT-universal sets to get a $\poly(D)$-round distributed algorithm for exact computation of small cuts. The only randomized part of the algorithm above is in defining the $\ell=\poly(D)$ subgraphs $G_i$. Instead of sampling these subgraphs, each vertex computes them explicitly and locally. First, we rename all the edges to be in $[1,m]$. This can be easily done in $O(D)$ rounds. 
Now, each vertex locally computes a family of universal sets for parameters $m,k=O(\lambda \cdot D),q=\lambda$.
By Lemma \ref{lem:universalsets}, the family $\mathcal{S}$ contains $(\lambda \cdot D)^{\lambda}=\poly(D)$ subsets in $[1,m]$. Each of the sets $S_i \in \mathcal{S}$ will be used as a subgraph $G_i$ in the $i^{th}$ iteration. 
That is, we iterate over all subsets (subgraphs) in $\mathcal{S}$. In iteration $i$, all vertices know the set $S_i$ and thus can locally decide which of their incident edges is in $G_i$.
The correctness now follows the exact same line as that of the randomized algorithm.
%
\section{Computation of All Edge Connectivities}\vspace{-5pt}
Finally, we consider the more general task of computing the edge connectivity of all graph edges up to some constant value $\lambda$. For an edge $e=(u,v)$, let $\lambda(e)$ be the $u$-$v$ edge connectivity in $G$, that is, the number of edge-disjoint $u$-$v$ paths in $G$. By using the recent notion of low-congestion cycle cover \cite{parter-yogevSODAa}, we show:
\begin{lemma}\label{thm:dist-edge-connect}[Distributed All Edge Connectivities]
For every $D$-diameter graph $G$, there is a randomized distributed algorithm that w.h.p. computes all edge connectivities up to some constant value $\lambda$ within $2^{O(\sqrt{\log n})}\cdot \poly(D)$ rounds. That is, in the output solution, the endpoints of every edge $e=(u,v)$ know the connectivity $\lambda(e)$ of this edge, as well as a certificate for that connectivity.
\end{lemma}

\paragraph{Low Congestion Cycle Covers.}
A $(\dilation,\congestion)$ cycle cover $\mathcal{C}$ is a collection of cycles of length at most $\dilation$, such that each edge appears on at least one cycle and at most $\congestion$ cycles. We will use the recent deterministic distributed construction of cycle covers of \cite{parter-yogevCycles}.

\begin{lemma}\label{thm:distcyclecover}\cite{parter-yogevCycles}[Distributed Cycle Cover]
For every bridgeless $n$-vertex graph $G=(V,E)$ with diameter $D$, one can compute a
$(\dilation,\congestion)$ cycle cover 
$\cC$ with $\dilation=2^{O(\sqrt{\log n})}\cdot D$ and $\congestion=2^{O(\sqrt{\log n})}$, within 
 $\widetilde{O}(\dilation \cdot \congestion)$ rounds. 
\end{lemma}

Combining the lemma above with the centralized construction of nearly-optimal cycle covers of \cite{parter-yogevSODAa}, we get:
\begin{corollary}\label{thm:distoptcyclecover}[Distributed Opt. Cycle Cover]
For an $n$-vertex graph $G=(V,E)$ (not necessarily connected), there is a randomized algorithm $\ApproxGraphCover$ that given the graph $G$ and parameter $D'$ computes w.h.p. a cycle collection $\cC$ such that: (a) every edge $e$ that lies on a cycle $C_e$ in $G$ of length at most $D'$ is covered by a cycle $C' \in \cC$ of length $2^{O(\sqrt{\log n})}\cdot |C_e|$, and (b) each edge appears on $2^{O(\sqrt{\log n})}$ cycles. In the output format of the algorithm, every edge $e$ learns all the cycles in $\cC$ that go through this edge. 
The round complexity of Alg. $\ApproxGraphCover$ is $2^{O(\sqrt{\log n})}\cdot D'$.
\end{corollary}
The high level idea of Alg. $\ApproxGraphCover$ is based on the notion of neighborhood covers. Roughly speaking, the $k$ neighborhood cover for a graph $G$ is a collection of subgraphs $G_1,\ldots, G_\ell$ such that the following three properties hold: (1) for each vertex $v$, there is a subgraph $G_i$ that contains its entire $k$-hop neighborhood, (2) the diameter of each graph $G_i$ is $O(k\log n)$, and (3) each vertex appears on $O(\log n)$ subgraphs. 
Alg. $\ApproxGraphCover$ is obtained by applying the cycle cover algorithm of Theorem \ref{thm:distcyclecover} on each subgraph $G_i$ in the $k$ neighborhood cover of $G$, for every value $k=2^{j}$, $j \in \{1,\lceil \log(D') \rceil\}$. This increases the total congestion of the cycles by at most a logarithmic factor. 
\indent To see why this works, consider an edge $e$ 
that lies on a cycle $C_e$ of length $|C_e|\leq D'$ in $G$. Letting $|C_e|\in [2^{j-1},2^{j}]$, we have that $C_e$ is fully contained in one of the subgraphs of a $k$ neighborhood cover for $k=2^j$. Hence, due to Theorem \ref{thm:distcyclecover} the edge $e$ is covered by a cycle of length $2^{O(\sqrt{\log n})}|C_e|$ as desired. 
%

\paragraph{From Cycle Covers to Edge Connectivities.}
The algorithm for computing all the edge connectivities is based on combining the FT-sampling approach with Alg. $\ApproxGraphCover$. In the high level, using the sampling technique, the algorithm attempts to compute not a single cycle for covering an edge $e=(u,v)$, but rather a collection of $\lambda$ edge-disjoint cycles that covers this edge (i.e., the edge $e$ is the only common edge in these cycles). If it fails in finding these edge disjoint cycles, it deduces that the $u$-$v$ connectivity is less than $\lambda$. In the latter case, it also finds all $u$-$v$ cuts in $G$.

Let $D'=2c\cdot \lambda \cdot D+1$. 
The algorithm consists of $\ell=O(\lambda \cdot D^{\lambda} \log n)$ iterations that we treat as \emph{experiments}. In each experiment $i$, we sample each edge $e \in E(G)$ into $G_i$ with probability $p=(1-1/(D')^{\lambda})$, and compute a cycle cover $\mathcal{C}_i$ by applying Alg. $\ApproxGraphCover$ on $G_i$ with parameter $D'$. 
For every edge $e=(u,v)$, let $G_{u,v}=\bigcup_{i}\{C ~\mid~ e \in C, C\in \mathcal{C}_i\}$ be the union of all cycles that go through $e$. The nodes $u,v$ compute the edge connectivity of $e$ by \emph{locally} computing the $u$-$v$ cut in the subgraph $G_{u,v}$. For a pseudo-code of the algorithm see Fig. \ref{fig:main-algorithm-non-tree-intro}. 
By a similar argument to that of Cl. \ref{cl:succprobcut}, we show:
\begin{claim}\label{cl:cert-edge-conn}
The subgraph $G_{u,v}$ is a $u$-$v$ connectivity certificate up to connectivity of $\lambda$.
\end{claim}
\begin{proof}
Let $e=(u,v)$ such that $u$ and $v$ are $\lambda$-edge connected. In other words, $u$ and $v$ are $(\lambda-1)$-connected in $G\setminus \{e\}$. Note that since the diameter of $G \setminus \{e\}$ is at most $2D$, by 
Lemma \ref{lem:bounded-length-connected}, for every $F \subseteq E(G) \setminus \{e\}$, $|F|\leq \lambda-2$, we have that $\dist(u,v, G \setminus (F \cup \{e\}))\leq 2c\cdot \lambda \cdot D$. Therefore, for any $F \subseteq E(G) \setminus \{e\}$, $|F|\leq \lambda-2$ the subgraph $G \setminus F$ contains a cycle that covers $e$ of length at most $D'=2c\cdot \lambda \cdot D+1$. 

To show that $G_{u,v}$ is a $\lambda$-certificate for such a neighboring pair $u,v$ (that is $\lambda$ edge connected in $G$), it is sufficient to show that for every failing of at most $\lambda-2$ edges $F$ where $e \notin F$, $G_{u,v}$ contains a $u$-$v$ path in $G_{u,v}\setminus (F \cup \{e\})$. Or in the other words, that $G_{u,v}\setminus F$ contains a cycle that covers $e$. 

Fix a failing set $F$, where $e \notin F$ and $|F|\leq \lambda-2$, and $u$ and $v$ are connected in $G \setminus (F \cup \{(u,v)\})$. We say that iteration $i$ is successful for such a triplet $\langle u,v, F \rangle$ if $F \cap G_i=\emptyset$, $(u,v)\in G_i$ and $\pi(u,v,G\setminus (F \cup \{(u,v)\})) \subseteq G_i$. Note that in such a case, since $G_i$ contains a cycle of length at most $D'$ that covers $e$, 
Algorithm $\ApproxGraphCover$ computes a cycle $C \subseteq G_i$ that covers the edge $e=(u,v)$, and thus a $u$-$v$ path in $G_i \setminus (F \cup \{e\})$ as desired.  
It remains to show that w.h.p. every triplet $\langle u,v, F \rangle$ has at least one successful iteration.
Since each edge is sampled w.p. $p$ into $G_i$, the iteration is successful with probability $\Omega(1/D^{\lambda})$. By simple application of the Chernoff bound, we get that the probability that a given triplet $\langle u,v, F \rangle$  does not have a successful iteration is at most $1/n^{c' \cdot \lambda}$. Thus by applying the union bound over all $n^{\lambda+2}$ triplets, the claim follows. 
\end{proof}
The proof of Lemma \ref{thm:dist-edge-connect} follows by Cor. \ref{thm:distoptcyclecover} and Cl. \ref{cl:cert-edge-conn}. Note that this algorithm can be made deterministic while keeping the same round complexity, by using the derandomization of the FT-sampling approach from Sec. \ref{sec:det-cut} along with the deterministic neighborhood cover construction of \cite{GhaffariK18}. 
\begin{lemma}\label{lem:edge-con-det}
All edge connectivities, up to a constant $\lambda$, can be computed deterministically in $2^{\sqrt{\log n \cdot \log\log n}}\cdot \poly(D)$.
\end{lemma}
\begin{proof}
The algorithm requires two main adaptations. First the randomized algorithm $\ApproxGraphCover$ is made deterministic by using the deterministic construction of neighborhood covers of \cite{GhaffariK18} that uses $2^{\sqrt{\log n \cdot \log\log n}}$ rounds. 
%
%
%
In the second part, we use the derandomization of the FT-sampling, as in the minimum cut algorithm. 
Overall, the total round complexity is $2^{O(\sqrt{\log n\cdot \log\log n})} \cdot \poly(D)$.
\end{proof}
\begin{figure}[!h]
	\begin{boxedminipage}{\textwidth}
		\vspace{3mm} \textbf{Algorithm $\DistEdgeConn(G=(V,E),\lambda)$}
		\begin{enumerate}
    \item For $i=1$ to $O(\lambda \cdot D)^{2\lambda}$:
		\begin{itemize}
		\item Sample each edge $e$ into $G_i$ w.p. $p=(1-1/D')$.
		\item $\mathcal{C}_i \gets \ApproxGraphCover(G_i,D')$.
		\end{itemize}
		\item $\mathcal{C}=\bigcup_{i} \mathcal{C}_i$. 
		\item For every edge $e=(u,v)$, let $G_{u,v}=\{ C \in \mathcal{C} ~\mid~ e \in C\}$. 
		\item $\lambda(e)=\MinCut(G_{u,v})$. 
		\end{enumerate}
	\end{boxedminipage}
	\caption{Distributed Computation of $\lambda$-edge connected cycle covers}
	\label{fig:main-algorithm-non-tree-intro}
\end{figure}

\section{Sparse Connectivity Certificates}
Finally, we consider the related problem of computing sparse connectivity certificates.
We start by observing that an FT-spanner for the graph is also a connectivity certificate. 

\paragraph{Fault Tolerant Spanners.}
Fault tolerant (FT) spanners \cite{levcopoulos1998efficient,chechik2010fault} are sparse subgraphs that preserve pairwise distances in $G$ (up to some multiplicative stretch) even when several edges or vertices in the graph fails. These spanners have been introduced by Levcopoulos for geometric graphs \cite{levcopoulos1998efficient}, and later on by Chechik et al. \cite{chechik2010fault} for general graphs. 

\begin{definition}[Fault Tolerant Spanners]
For positive integers $k,f$, an $f$ edge fault-tolerant $(2k-1)$ spanner for an $n$-vertex graph $G=(V,E)$ is a subgraph $H \subseteq G$ satisfying that $\dist(u,v,H \setminus F) \leq (2k-1) \dist(u,v,G \setminus F)$ for every $u, v \in V$ and $F \subseteq E$, $|F|\leq f$.
Vertex fault-tolerant spanners are defined analogously where the fault $F \subset V$.
\end{definition}\vspace{-3pt}
Chechik et al. \cite{chechik2010fault} gave a generic algorithm for computing these spanners against edge failures. 
\vspace{-3pt}\begin{fact}\label{fc:ftspanner}\cite{chechik2010fault}
Let $\cA$ be an algorithm for computing the standard (fault-free) $(2k-1)$ spanner with $O(n^{1+1/k})$ edges and time $t$. Then one can compute $f$ edge fault-tolerant $(2k-1)$ spanners with $O(f \cdot n^{1+1/k})$ edges in time $O(f \cdot t)$.
\end{fact} 
Dinitz and Krauthgamer provided a similar transformation for vertex faults that is based on the FT-sampling technique, see Thm. 2.1 of \cite{dinitz2011fault}.

\paragraph{Certificates from Fault Tolerant Spanners.}
The relation between FT-spanners and connectivity certificates is based on the following observation.
\vspace{-7pt}\begin{observation}\label{obs:ftsc}
An $f$ \emph{edge} (resp., vertex) FT spanner $H \subseteq G$ is a certificate for the $f$ \emph{edge} (resp., vertex) connectivity of the graph.
\end{observation}
\begin{proof}[Proof of Observation \ref{obs:ftsc}]
Consider a $\lambda$-edge connected graph $G$, and let $H$ be an $f$-FT-spanner for $G$ with $f=\lambda-1$. 
We show that $H$ is $\lambda$-edge connected, by showing that for every vertex pair $s,t$ and a subset of $F$ edge faults, $|F|\leq f$, there exists an $s$-$t$ path in $H \setminus F$. 
Let $P$ be an $s$-$t$ path in $G \setminus F$. Since $G$ is $\lambda$-connected, such a path exists. For every edge $e=(u,v) \in P \setminus H$, it holds that $\dist(u,v,H \setminus F)\leq 2k-1$, thus in particular, every neighboring pair $u,v$ on $P$ are connected in $H \setminus F$, the claim follows. The proof of $\lambda$-vertex connected graphs works in the same manner.
%
\end{proof}
Since spanners with logarithmic stretch have linear size, fault tolerant spanners with logarithmic stretch are sparse connectivity certificates. By using the existing efficient (in fact local) distributed algorithms for spanners, we get:
\begin{observation}\label{obs:sparse-det}
(1) There exists a randomized algorithm for computing a sparse $\lambda$-edge certificate with $O(\lambda n)$ edges within $O(\lambda \cdot \log^{1+o(1)} n)$ rounds, w.h.p.;
(2) There exists a randomized algorithm for computing a $\lambda$-vertex certificate for the $\lambda$-vertex connectivity with $O(\lambda^2 \cdot \log n)$ edges within $O(\lambda^3 \cdot \log^{2+o(1)} n)$ rounds, w.h.p.;
\end{observation}
\begin{proof}
Claim (1) follows by combining the ultra-sparse spanners construction of Pettie \cite{Pettie10} with Fact \ref{fc:ftspanner}.
To obtain sparse certificates, we use Fact \ref{fc:ftspanner} with algorithm $\cA$ taken to be the ultra-sparse algorithm by Pettie \cite{Pettie10}. This algorithm computes an $O(2^{\log^* n} \log n)$-spanner $H$ with $O(n)$ edges within $O(\log^{1+o(1)} n)$ rounds. By taking $\lambda$ disjoint copies of such spanner, constructed sequentially one after the other, we obtain a certificate with $O(\lambda \cdot n)$ edges. 
In the same manner, Claim (2) follows by plugging the algorithm of Pettie \cite{Pettie10} in the meta algorithm for FT-spanners against vertex faults by Dinitz and Krauthgamer (Thm. 2.1 in \cite{dinitz2011fault}).
\end{proof}
While Obs. \ref{obs:sparse-det} gives efficient algorithm when $\lambda=O(1)$, it is less efficient for large connectivity values. In the next lemma we apply the edge sampling technique of Karger \cite{karger1999random} to omit the dependency in $\lambda$ in the round complexity. Ideas along this line appear in \cite{ghaffari2013distributed}.
\vspace{-3pt}\begin{lemma}
For every $\lambda=\Omega(\log n)$, and $\epsilon \in [0,1]$, given an $\lambda$-edge connected graph $G$, one can compute, w.h.p., an $(1-\epsilon)\lambda$-edge sparse certificate within $O(1/\epsilon^2 \cdot \log^{2+o(1)} n)$ rounds.
\end{lemma}
\begin{proof}
We restrict attention to $\lambda=\Omega(\log n)$, as otherwise, the lemma follows immediately from Obs. \ref{obs:sparse-det}. 
The key idea for omitting the dependency in $\lambda$ is by randomly decomposing the graph into spanning subgraphs each with connectivity $\min\{\lambda, \Theta(\log n /\epsilon^2)\}$ using random edge sampling, and to run the algorithm of Obs. \ref{obs:sparse-det} on each of the subgraphs.
We randomly put each edge of $G$ in one of $\mu$ subgraphs $G_1,\ldots, G_\mu$ for $\mu=\lceil \lambda \cdot \epsilon^2/(20\log n) \rceil$. Karger \cite{karger1999random} showed that each subgraph $G_i$ has edge-connectivity in $(1\pm \epsilon)\lambda/\mu$ with high probability. In addition, the summation of the edge-connectivities $\lambda_1, \ldots, \lambda_\mu$ of the subgraphs $G_1,\ldots, G_{\mu}$ is at least $\lambda(1-\epsilon)$.
The algorithm then computes a $\lambda'$-edge sparse certificate $H_i$ for $\lambda'=(1-\epsilon)\lambda/\mu$ in each $G_i$ subgraph simultaneously. The output certificate $H$ is the union of all $H_i$ subgraphs. 

We next analyze the construction, and start with round complexity.
Since the $G_i$ subgraphs are edge disjoint, applying the algorithm of Obs. \ref{obs:sparse-det} takes $O(\lambda' \cdot \log^{1+o(1)}n)$ rounds which is $O(\log^{2+o(1)}n)$ rounds.
The edge bound follows immediately as $|E(H)|=c \cdot\mu \cdot \lambda' n=O(\lambda n)$ as required. 
It remains to show that $H$ is indeed a $(1-\epsilon)\lambda$-edge connectivity certificate. 
Consider a pair of vertices $s,t$, and a sequence of at most $(1-\epsilon)\lambda$ edge faults $F$. We will show that $s$ and $t$ are connected in $H \setminus F$. 
Since the $\mu$ subgraphs are edge-disjoint, there must be a subgraph $G_i$ containing at most $\lambda'=(1-\epsilon)\lambda/\mu$ of the faults. Let $F_i=F \cap G_i$.
Since $G_i$ is $\lambda'$-edge connected, $s$ and $t$ are connected in $G_i \setminus F_i$. Since $H_i$ is a $\lambda'$-edge certificate of $G_i$, it also holds that $s$ and $t$ are connected in $H_i \setminus F_i$ and thus also in $H \setminus F$ (i.e., as by definition, $(F \setminus F_i) \cap H_i=\emptyset$).
%
%
The claim follows. 
\end{proof}



%

\newpage
\bibliographystyle{alpha} 
\bibliography{crypto}

\newcommand{\etalchar}[1]{$^{#1}$}
\begin{thebibliography}{10}

\bibitem{alon1986explicit}
Noga Alon.
\newblock Explicit construction of exponential sized families of k-independent
  sets.
\newblock {\em Discrete Mathematics}, 58(2):191--193, 1986.

\bibitem{AlonCC19}
Noga Alon, Shiri Chechik, and Sarel Cohen.
\newblock Deterministic combinatorial replacement paths and distance
  sensitivity oracles.
\newblock In {\em 46th International Colloquium on Automata, Languages, and
  Programming, {ICALP} 2019, July 9-12, 2019, Patras, Greece.}, pages
  12:1--12:14, 2019.

\bibitem{Censor-HillelD17}
Keren Censor{-}Hillel and Michal Dory.
\newblock Fast distributed approximation for {TAP} and 2-edge-connectivity.
\newblock In {\em 21st International Conference on Principles of Distributed
  Systems, {OPODIS} 2017, Lisbon, Portugal, December 18-20, 2017}, pages
  21:1--21:20, 2017.

\bibitem{chechik2010fault}
Shiri Chechik, Michael Langberg, David Peleg, and Liam Roditty.
\newblock Fault tolerant spanners for general graphs.
\newblock {\em SIAM Journal on Computing}, 39(7):3403--3423, 2010.

\bibitem{DagaSTOC19}
Mohit Daga, Monika Henzinger, Danupon Nanongkai, and Saranurak.
\newblock Distributed edge connectivity in sublinear time.
\newblock In {\em STOC}, 2019.

\bibitem{dinitz2011fault}
Michael Dinitz and Robert Krauthgamer.
\newblock Fault-tolerant spanners: better and simpler.
\newblock In {\em Proceedings of the 30th annual ACM SIGACT-SIGOPS symposium on
  Principles of distributed computing}, pages 169--178. ACM, 2011.

\bibitem{Dory18}
Michal Dory.
\newblock Distributed approximation of minimum k-edge-connected spanning
  subgraphs.
\newblock In {\em Proceedings of the 2018 {ACM} Symposium on Principles of
  Distributed Computing, {PODC} 2018, Egham, United Kingdom, July 23-27, 2018},
  pages 149--158, 2018.

\bibitem{EvenIR98}
Shimon Even, Gene Itkis, and Sergio Rajsbaum.
\newblock On mixed connectivity certificates.
\newblock {\em Theor. Comput. Sci.}, 203(2):253--269, 1998.

\bibitem{friedman1984constructing}
J~Friedman.
\newblock Constructing o (n log n) size monotone formulae for the k-th
  elementary symmetric polynomial of n boolean variables.
\newblock In {\em Foundations of Computer Science, 1984. 25th Annual Symposium
  on}, pages 506--515. IEEE, 1984.

\bibitem{GhaffariThesis17}
Mohsen Ghaffari.
\newblock {\em Improved Distributed Algorithms for Fundamental Graph Problems}.
\newblock PhD thesis, MIT, {USA}, 2017.
\newblock URL:
  \url{https://groups.csail.mit.edu/tds/papers/Ghaffari/PhDThesis-Ghaffari.pdf}.

\bibitem{ghaffari2013distributed}
Mohsen Ghaffari and Fabian Kuhn.
\newblock Distributed minimum cut approximation.
\newblock In {\em International Symposium on Distributed Computing}, pages
  1--15. Springer, 2013.

\bibitem{ghaffari2013distributed-Arxiv}
Mohsen Ghaffari and Fabian Kuhn.
\newblock Distributed minimum cut approximation.
\newblock {\em arXiv preprint arXiv:1305.5520}, 2013.

\bibitem{GhaffariK18}
Mohsen Ghaffari and Fabian Kuhn.
\newblock Derandomizing distributed algorithms with small messages: Spanners
  and dominating set.
\newblock In {\em 32nd International Symposium on Distributed Computing, {DISC}
  2018, New Orleans, LA, USA, October 15-19, 2018}, pages 29:1--29:17, 2018.

\bibitem{karger1999random}
David~R Karger.
\newblock Random sampling in cut, flow, and network design problems.
\newblock {\em Mathematics of Operations Research}, 24(2):383--413, 1999.

\bibitem{levcopoulos1998efficient}
Christos Levcopoulos, Giri Narasimhan, and Michiel Smid.
\newblock Efficient algorithms for constructing fault-tolerant geometric
  spanners.
\newblock In {\em Proceedings of the thirtieth annual ACM symposium on Theory
  of computing}, pages 186--195. ACM, 1998.

\bibitem{nagamochi1992linear}
Hiroshi Nagamochi and Toshihide Ibaraki.
\newblock A linear-time algorithm for finding a sparsek-connected spanning
  subgraph of ak-connected graph.
\newblock {\em Algorithmica}, 7(1-6):583--596, 1992.

\bibitem{nanongkai2014almost}
Danupon Nanongkai and Hsin-Hao Su.
\newblock Almost-tight distributed minimum cut algorithms.
\newblock In {\em International Symposium on Distributed Computing}, pages
  439--453. Springer, 2014.

\bibitem{parter-yogevSODAa}
Merav Parter and Eylon Yogev.
\newblock Low congestion cycle covers and their applications.
\newblock {\em SODA}, 2019.

\bibitem{parter-yogevCycles}
Merav Parter and Eylon Yogev.
\newblock Optimal short cycle decomposition in almost linear time.
\newblock {\em ICALP}, 2019.

\bibitem{Peleg:2000}
David Peleg.
\newblock {\em Distributed Computing: A Locality-sensitive Approach}.
\newblock SIAM, 2000.

\bibitem{Pettie10}
Seth Pettie.
\newblock Distributed algorithms for ultrasparse spanners and linear size
  skeletons.
\newblock {\em Distributed Computing}, 22(3):147--166, 2010.
\newblock URL: \url{https://doi.org/10.1007/s00446-009-0091-7}, \href
  {http://dx.doi.org/10.1007/s00446-009-0091-7}
  {\path{doi:10.1007/s00446-009-0091-7}}.

\bibitem{pritchard2011fast}
David Pritchard and Ramakrishna Thurimella.
\newblock Fast computation of small cuts via cycle space sampling.
\newblock {\em ACM Transactions on Algorithms (TALG)}, 7(4):46, 2011.

\bibitem{detND-Arxiv}
V{\'{a}}clav Rozhon and Mohsen Ghaffari.
\newblock Polylogarithmic-time deterministic network decomposition and
  distributed derandomization.
\newblock {\em arXiv preprint arXiv:1907.10937}, 2019.

\bibitem{sarma2012distributed}
Atish~Das Sarma, Stephan Holzer, Liah Kor, Amos Korman, Danupon Nanongkai,
  Gopal Pandurangan, David Peleg, and Roger Wattenhofer.
\newblock Distributed verification and hardness of distributed approximation.
\newblock {\em SIAM Journal on Computing}, 41(5):1235--1265, 2012.

\bibitem{thurimella1997sub}
Ramakrishna Thurimella.
\newblock Sub-linear distributed algorithms for sparse certificates and
  biconnected components.
\newblock {\em Journal of Algorithms}, 23(1):160--179, 1997.

\bibitem{TCS-010}
Salil~P. Vadhan.
\newblock Pseudorandomness.
\newblock {\em Foundations and Trends® in Theoretical Computer Science},
  7(1–3):1--336, 2012.
\newblock URL: \url{http://dx.doi.org/10.1561/0400000010}, \href
  {http://dx.doi.org/10.1561/0400000010} {\path{doi:10.1561/0400000010}}.

\bibitem{wang2003complex}
Xiao~Fan Wang and Guanrong Chen.
\newblock Complex networks: small-world, scale-free and beyond.
\newblock {\em IEEE circuits and systems magazine}, 3(1):6--20, 2003.

\bibitem{weimann2013replacement}
Oren Weimann and Raphael Yuster.
\newblock Replacement paths and distance sensitivity oracles via fast matrix
  multiplication.
\newblock {\em ACM Transactions on Algorithms (TALG)}, 9(2):14, 2013.

\end{thebibliography}

\end{document}